\documentclass{llncs}[orivec]

\usepackage{amssymb}
\usepackage{latexsym}
\usepackage{amsfonts,amsmath}
\usepackage{alltt}
\usepackage{epic,epsfig}
\usepackage{graphics}
\usepackage{algorithm}
\usepackage{algpseudocode}
\usepackage{algpseudocodex}
\usepackage{algorithm}
\usepackage{xcolor}
\usepackage{float}
\usepackage{tikz}
\usetikzlibrary{arrows.meta, positioning}

\usepackage{environ}
\usepackage{etoolbox}

\newtoggle{festschrift}
\newtoggle{jolli}
\newif\ifhideproofs

\togglefalse{festschrift}
\toggletrue{jolli}
\hideproofsfalse

\ifhideproofs
  \usepackage{environ}
  \NewEnviron{hide}{}

\fi
\NewEnviron{maybePrint}[1]{%
  \iftoggle{#1}{\BODY}{}%
    }

\newcommand{\lgw}{| \hspace{-0.18mm}w \hspace{-0.18mm}|}

\newcommand{\lgN}{| \hspace{-0.18mm}N \hspace{-0.18mm}|}
\newcommand{\lgx}[1]{| \hspace{-0.18mm}{#1} \hspace{-0.18mm}|}

\newcommand{\card}{\mathrm{Card}}
\newcommand{\Ba}{\ensuremath{\Box_a}}
\newcommand{\Bb}{\ensuremath{\Box_b}}
\newcommand{\Ra}{{Ra}}
\newcommand{\Rb}{{Rb}}
\newcommand{\Kdeab}{\ensuremath{\mathbf{KDe_{a,b}}}}
\newcommand{\Kde}{\ensuremath{\mathbf{KDe}}}

\newcommand{\K}{\ensuremath{\mathbf{K}}}
\newcommand{\SQ}{\ensuremath{\mathbf{S4}}}
\newcommand{\KT}{\ensuremath{\mathbf{KT}}}
\newcommand{\nat}{\mathbb{N}}

\newcommand{\pipe}{\hspace{-0.21mm}|\hspace{-0.21mm}}

\newcommand{\inc}{\subseteq}
\newcommand{\uplim}{\ensuremath{2^{c_{0}.(d(w)+1).\lgw}}}
\newcommand{\wpp}{\ensuremath{\tilde{w}}}

\newcommand{\Tpp}{\ensuremath{\tilde{T}}}
\newcommand{\nextW}{\mbox{\texttt{NextW}}}
\newcommand{\satW}{\mbox{\texttt{SatW}}}
\newcommand{\chooseW}{\mbox{\texttt{ChooseW}}}
\newcommand{\chooseCCS}{\mbox{\texttt{ChooseCCS}}}
\newcommand{\sat}{\mbox{\texttt{Sat}}}

\newcommand{\CCS}{\mbox{\texttt{CCS}}}
\newcommand{\SF}{\mbox{\texttt{SF}}}
\newcommand{\CSF}{\mbox{\texttt{CSF}}}

\newcommand{\true}{\mbox{\texttt{True}}}

\newcommand{\all}{\mbox{\texttt{all}}}
\newcommand{\algand}{\mbox{\texttt{and}}}

\usepackage{amsmath}

\def\PSPACE{\mathbf{PSPACE}}
\def\card{\mathtt{Card}}
\def\NPSPACE{\mathbf{NPSPACE}}
\def\EXPTIME{\mathbf{EXPTIME}}
\def\NEXPTIME{\mathbf{NEXPTIME}}
\def\coNEXPTIME{\mathbf{coNEXPTIME}}
\def\CPL{\mathbf{CPL}}
\def\K{\mathbf{K}}
\def\Log{\mathbf{L}}
\def\At{\mathbf{At}}
\def\Fo{\mathbf{Fo}}
\def\N{{\mathbb{N}}}
\def\Axiom{\mathbf{A}}
\def\Rule{\mathbf{R}}
\mathchardef\mhyphen="2D
\begin{document}

\title{
\begin{maybePrint}{festschrift}
Complexity of some modal logics of density\end{maybePrint}
\begin{maybePrint}{jolli}
Complexity of some modal logics of density (extended version) \end{maybePrint}}

\author{Philippe Balbiani \and Olivier Gasquet}
%
%
\institute{
Institut de recherche en informatique de Toulouse
\\
CNRS-INPT-UT
}
\maketitle
\begin{abstract}
By using a selective filtration argument, we prove that the satisfiability problem of the unimodal logic of density is in $\EXPTIME$.
By using a tableau-like approach, we prove that the satisfiability problem of the bimodal logic of weak density is in $\PSPACE$. 
\end{abstract}

\section*{Introduction}
For modal logics defined by \emph{grammar axioms} of the form $\langle a_1\rangle\ldots\langle a_m\rangle p\rightarrow \langle b_1\rangle\ldots\langle b_n\rangle p$ the satisfiability problem is known to be undecidable in general~\cite{FarinasdelCerro1988-FARGL} while for some specific \emph{grammar logics}, the satisfiability problem is known to be decidable in $\NEXPTIME$~\cite{Lyon24} and even in $\PSPACE$, for instance the well-known $\K$,$\KT$, $\K 4$ or $\SQ$, or also a logic like $\K+\lozenge p\leftrightarrow\lozenge\lozenge p$~\cite{FAR-GAS99}.
\\
\\
In this paper, we study the complexity of some modal logics defined by axioms of the form $\langle a\rangle p\rightarrow\langle b\rangle\langle c\rangle p$.
Firstly, by using a selective filtration argument, we prove that the satisfiability problem of the unimodal logic of density is in $\EXPTIME$.
Secondly, by using a tableau-like approach, we prove that the satisfiability problem of the bimodal logic of weak density is in $\PSPACE$.
%
%
%
%
%
%
Let $\Kde$ be the modal logic $\K+\lozenge p\rightarrow\lozenge\lozenge p$.
In Section~\ref{section:lower:bound:complexity:of:KDe}, we prove that $\Kde$ is $\PSPACE$-hard.
In Section~\ref{section:complexity:of:KDe}, we prove that $\Kde$ is in $\EXPTIME$. Then, let $\Kdeab$ be the modal logic $\K+\Diamond_a p\rightarrow \Diamond_a\Diamond_b p$, the rest of the paper is devoted to prove it to be $\PSPACE$-complete. 
\paragraph{Syntax}
Let $\At$ be the set of all atoms $(p,q,\ldots)$.
The set $\Fo$ of all formulas $(\phi,\psi,\ldots)$ is defined by
\[\phi:=p\mid\bot\mid\neg\phi\mid(\phi\wedge\phi)\mid\square\phi\]
where $p$ ranges over $\At$.
We follow the standard rules for omission of the parentheses.
We use the standard abbreviations for the Boolean connectives $\top$, $\vee$ and $\rightarrow$.
The {\em degree} of a formula $\phi$ (in symbols $d(\phi)$) is defined as usual.
For all formulas $\phi$, $\pipe\phi\pipe$ denotes the number of occurrences of symbols in $\phi$.
For all formulas $\phi$, we write $\lozenge\phi$ as an abbreviation instead of $\neg\square\neg\phi$.
\paragraph{Semantics}
A {\em frame}\/ is a couple $(W,R)$ where $W$ is a nonempty set and $R$ is a binary relation on $W$.
A frame $(W,R)$ is {\em dense}\/ if for all $s,t\in W$, if $sRt$ then there exists $u\in W$ such that $sRu$ and $uRt$.
A {\em valuation on a frame $(W,R)$}\/ is a function $V\ :\ \At\longrightarrow\wp(W)$.
A {\em model}\/ is a $3$-tuple consisting of the $2$ components of a frame and a valuation on that frame.
A {\em model based on the frame $(W,R)$}\/ is a model of the form $(W,R,V)$.
With respect to a model $(W,R,V)$, for all $s\in W$ and for all formulas $\phi$, the {\em satisfiability of $\phi$ at $s$ in $(W,R,V)$}\/ (in symbols $s\models\phi$) is inductively defined as usual.
In particular,
\begin{itemize}
\item $s\models\square\phi$ if and only if for all $t\in W$, if $sRt$ then $t\models\phi$.
\end{itemize}
As a result,
\begin{itemize}
\item $s\models\lozenge\phi$ if and only if there exists $t\in W$ such that $sRt$ and $t\models\phi$.
\end{itemize}
A formula $\phi$ is {\em true in a model $(W,R,V)$}\/ (in symbols $(W,R,V)\models\phi$) if for all $s\in W$, $s\models\phi$.
A formula $\phi$ is {\em valid in a frame $(W,R)$}\/ (in symbols $(W,R)\models\phi$) if for all models $(W,R,V)$ based on $(W,R)$, $(W,R,V)\models\phi$.
A formula $\phi$ is {\em valid in a class ${\mathcal C}$ of frames}\/ (in symbols ${\mathcal C}\models\phi$) if for all frames $(W,R)$ in ${\mathcal C}$, $(W,R)\models\phi$.
\paragraph{A decision problem}
Let $DP$ be the following decision problem:
\begin{description}
\item[input:] a formula $\phi$,
\item[output:] determine whether $\phi$ is valid in the class of all dense frames.
\end{description}
Using the fact that the least filtration of a dense model is dense, one may readily prove that $DP$ is in $\coNEXPTIME$~\cite[Chapter~$2$]{Blackburn:deRijke:Venema}.
We will prove in Section~\ref{section:complexity:of:KDe} that $DP$ is in $\EXPTIME$.
\paragraph{Axiomatization}
In our language, a {\em modal logic}\/ is a set of formulas closed under uniform substitution, containing the standard axioms of $\CPL$, closed under the standard inference rules of $\CPL$, containing the axioms
\begin{description}
\item[$(\Axiom1)$] $\square p\wedge\square q\rightarrow\square(p\wedge q)$,
\item[$(\Axiom2)$] $\square\top$,
\end{description}
and closed under the inference rule
\begin{description}
\item[$(\Rule1)$] $\frac{p\rightarrow q}{\square p\rightarrow\square q}$.
\end{description}
Let $\Kde$ be the least modal logic containing the formula $\square\square p\rightarrow\square p$.
As is well-known, $\Kde$ is equal to the set of all formulas $\phi$ such that $\phi$ is valid in the class of all dense frames.
This can be proved by using the so-called canonical model construction~\cite[Chapter~$4$]{Blackburn:deRijke:Venema}.
\paragraph{Theories}
Let $\Log$ be a modal logic.
A {\em $\Log$-theory}\/ is a set of formulas containing $\Log$ and closed under modus ponens.
A $\Log$-theory $\Gamma$ is {\em proper}\/ if $\bot\not\in\Gamma$.
A proper $\Log$-theory $\Gamma$ is {\em prime}\/ if for all formulas $\phi,\psi$, if $\phi\vee\psi\in\Gamma$ then either $\phi\in\Gamma$, or $\psi\in\Gamma$.
For all $\Log$-theories $\Gamma$ and for all sets $\Delta$ of formulas, let $\Gamma+\Delta$ be the $\Log$-theory $\{\psi\in\Fo\ :$ there exists $m\in\N$ and there exists $\phi_{1},\ldots,\phi_{m}\in\Delta$ such that $\phi_{1}\wedge\ldots\wedge\phi_{m}\rightarrow\psi\in\Gamma\}$.
For all $\Log$-theories $\Gamma$ and and for all formulas $\phi$, we write $\Gamma+\phi$ instead of $\Gamma+\{\phi\}$.
For all $\Log$-theories $\Gamma$, let $\square\Gamma$ be the $\Log$-theory $\{\phi\in\Fo\ :\ \square\phi\in\Gamma\}$.
\paragraph{Canonical model}
The {\em canonical frame of $\Log$}\/ is the couple $(W_{\Log},R_{\Log})$ where $W_{\Log}$ is the set of all prime $\Log$-theories and $R_{\Log}$ is the binary relation on $W_{\Log}$ such that for all $\Gamma,\Delta\in W_{\Log}$, $\Gamma R_{\Log}\Delta$ if and only if $\square\Gamma\subseteq\Delta$.
As is well-known, if $\Log$ contains $\Kde$ then $(W_{\Log},R_{\Log})$ is dense.
The {\em canonical valuation of $\Log$}\/ is the function $V_{\Log}\ :\ \At\longrightarrow\wp(W_{\Log})$ such that for all atoms $p$, $V_{\Log}(p)=\{\Gamma\in W_{\Log}:\ p\in\Gamma\}$.
The {\em canonical model of $\Log$}\/ is the triple $(W_{\Log},R_{\Log},V_{\Log})$.
The completeness of $\Kde$ is a direct consequence of {\bf (i)}~the fact that if $\Log$ contains $\Kde$ then $(W_{\Log},R_{\Log})$ is dense and {\bf (ii)}~the following Truth Lemma~\cite[Lemma~$4.21$]{Blackburn:deRijke:Venema}:
\begin{lemma}[Truth Lemma]
Let $\phi$ be a formula.
For all $\Gamma\in W_{\Log}$, $\phi\in\Gamma$ if and only if $(W_{\Log},R_{\Log},V_{\Log}),\Gamma\models\phi$.
\end{lemma}
\section{$DP$ is $\PSPACE$-hard}\label{section:lower:bound:complexity:of:KDe}
For all atoms $p$, let $\tau_{p}:\ \Fo\longrightarrow\Fo$ be the function inductively defined as follows:
\begin{itemize}
\item $\tau_{p}(q)=q$,
\item $\tau_{p}(\bot)=\bot$,
\item $\tau_{p}(\neg\phi)=\neg\tau_{p}(\phi)$,
\item $\tau_{p}(\phi\vee\psi)=\tau_{p}(\phi)\vee\tau_{p}(\psi)$,
\item $\tau_{p}(\square\phi)=\square(p\rightarrow\tau_{p}(\phi))$.
\end{itemize}
Obviously, for all atoms $p$ and for all $\phi\in\Fo$, $\pipe\tau_{p}(\phi)\pipe\leq5.\pipe\phi\pipe$.
\begin{lemma}\label{lemma:about:the:translation}
For all atoms $p$ and for all formulas $\phi$, if $p$ does not occur in $\phi$ then the following conditions are equivalent:
\begin{enumerate}
\item $\phi$ is valid in the class of all frames,
\item $\tau_{p}(\phi)$ is valid in the class of all frames,
\item $\tau_{p}(\phi)$ is valid in the class of all dense frames.
\end{enumerate}
\end{lemma}
\begin{proof}
Let $p$ be an atom and $\phi$ be a formula.
Suppose $p$ does not occur in $\phi$.
Obviously, $\mathbf{(2)\Rightarrow(3)}$.
Consequently, it suffices to prove that $\mathbf{(1)\Rightarrow(2)}$ and $\mathbf{(3)\Rightarrow(1)}$.
\\
$\mathbf{(1)\Rightarrow(2):}$
Suppose $\tau_{p}(\phi)$ is not valid in the class of all frames.
Hence, there exists a model $M=(W,R,V)$ and there exists $s\in W$ such that $M,s\not\models\tau_{p}(\phi)$.
Let $M^{\prime}=(W^{\prime},R^{\prime},V^{\prime})$ be the model such that
\begin{itemize}
\item $W^{\prime}=W$,
\item for all $t,u\in W$, $tR^{\prime}u$ if and only if $tRu$ and $u\in V(p)$,
\item for all atoms $q$, $V^{\prime}(q)=V(q)$.
\end{itemize}
As the reader may easily verify by induction on $\psi\in\Fo$, if $p$ does not occur in $\psi$ then for all $t\in W$, $M,t\models\tau_{p}(\psi)$ if and only if $M^{\prime},t\models\psi$.
Since $p$ does not occur in $\phi$ and $M,s\not\models\tau_{p}(\phi)$, then $M^{\prime},s\not\models\phi$.
Thus, $\phi$ is not valid in the class of all frames.
\\
$\mathbf{(3)\Rightarrow(1):}$
Suppose $\phi$ is not valid in the class of all frames.
Consequently, there exists a model $M=(W,R,V)$ and there exists $s\in W$ such that $M,s\not\models\phi$.
Without loss of generality, suppose $W$ and $R$ are disjoint.
Let $M^{\prime}=(W^{\prime},R^{\prime},V^{\prime})$ be the dense model such that
\begin{itemize}
\item $W^{\prime}=W\cup R$,
\item for all $t,u\in W$, $tR^{\prime}u$ if and only if $tRu$,
\item for all $t\in W$ and for all $(u,v)\in R$, $tR^{\prime}(u,v)$ if and only if $t=u$
\item for all $(t,u)\in R$ and for all $v\in W$, $(t,u)R^{\prime}v$ if and only if $u=v$,
\item for all $(t,u),(v,w)\in R$, $(t,u)R^{\prime}(v,w)$ if and only if $t=v$ and $u=w$,
\item for all atoms $q$, if $q\not=p$ then $V^{\prime}(q)=V(q)$.
\item $V^{\prime}(p)=W$.
\end{itemize}
As the reader may easily verify by induction on $\psi\in\Fo$, if $p$ does not occur in $\psi$ then for all $t\in W$, $M,t\models\psi$ if and only if $M^{\prime},t\models\tau_{p}(\psi)$.
Since $p$ does not occur in $\phi$ and $M,s\not\models\phi$, then $M^{\prime},s\not\models\tau_{p}(\phi)$.
Thus, $\tau_{p}(\phi)$ is not valid in the class of all dense frames.
\end{proof}
\begin{proposition}
$DP$ is $\PSPACE$-hard.
\end{proposition}
\begin{proof}
By Lemma~\ref{lemma:about:the:translation} and the fact that the validity problem in the class of all frames is $\PSPACE$-hard~\cite[Theorem~$6.50$]{Blackburn:deRijke:Venema}.
\end{proof}
\section{$DP$ is in $\EXPTIME$}\label{section:complexity:of:KDe}
From now on, the elements of $W_{\Kde}$~---~i.e. the prime $\Kde$-theories~---~will be denoted $s$, $t$, etc.
\\
\\
Let us consider a formula $\phi$.
Let $\Sigma_{\phi}$ be the set of all subformulas of $\phi$.
Let $n_{\phi}$ be the cardinal of $\Sigma_{\phi}$.
Obviously, $n_{\phi}\leq\pipe\phi\pipe$.
Let $(\psi_{1},\ldots,\psi_{n_{\phi}})$ be an enumeration of $\Sigma_{\phi}$ such that for all $i,j,k\in(n_{\phi})$,
\begin{itemize}
\item if $\psi_{i}=\neg\psi_{j}$ then $i>j$,
\item if $\psi_{i}=\psi_{j}\vee\psi_{k}$ then $i>j$ and $i>k$,
\item if $\psi_{i}=\square\psi_{j}$ then $i{>}j$.
\end{itemize}
A $\phi$-tip is an $n_{\phi}$-tuple $(a_{1},\ldots,a_{n_{\phi}})$ of bits such that for all $i,j,k\in(n_{\phi})$,
\begin{itemize}
\item if $\psi_{i}=\bot$ then $a_{i}=0$,
\item if $\psi_{i}=\neg\psi_{j}$ then $a_{i}=1-a_{j}$.
\item if $\psi_{i}=\psi_{j}\vee\psi_{k}$ then $a_{i}=\max\{a_{j},a_{k}\}$.
\end{itemize}
Obviously, there exists at most $2^{n_{\phi}}$ $\phi$-tips.
\\
\\
For all $s\in W_{\Kde}$, let $\tau_{\phi}(s)$ be the $n_{\phi}$-tuple $(a_{1},\ldots,a_{n_{\phi}})$ of bits such that for all $i\in(n_{\phi})$, if $\psi_{i}\in s$ then $a_{i}=1$ else $a_{i}=0$.
\begin{lemma}\label{lemma:3:about:phi:tips}
For all $s\in W_{\Kde}$, $\tau_{\phi}(s)$ is a $\phi$-tip.
\end{lemma}
\begin{proof}
This is an immediate consequence of the definitions.
\end{proof}
Let $(W_{\phi}^{0},R_{\phi}^{0})$ be the relational structure where
\begin{itemize}
\item $W_{\phi}^{0}$ is the set of all $\phi$-tips,
\item $R_{\phi}^{0}$ is the binary relation on $W_{\phi}^{0}$ such that for all $(a_{1},\ldots,a_{n_{\phi}}),(b_{1},\ldots,b_{n_{\phi}})\in W_{\phi}^{0}$, $(a_{1},\ldots,a_{n_{\phi}})R_{\phi}^{0}(b_{1},\ldots,b_{n_{\phi}})$ if and only if for all $i,j\in(n_{\phi})$, if $\psi_{i}=\square\psi_{j}$ and $a_{i}=1$ then $b_{j}=1$.
\end{itemize}
A $\phi$-clip is a relational structure $(W,R)$ where
\begin{itemize}
\item $W$ is a set included in $W_{\phi}^{0}$,
\item $R$ is a binary relation on $W$ included in $R_{\phi}^{0}$,
\item for all $s\in W_{\Kde}$, $\tau_{\phi}(s)\in W$,
\item for all $s,t\in W_{\Kde}$, if $sR_{\Kde}t$ then $\tau_{\phi}(s)R\tau_{\phi}(t)$.
\end{itemize}
Let ${\mathcal C}_{\phi}$ be the set of all $\phi$-clips.
Since there exists at most $2^{n_{\phi}}$ $\phi$-tips, then ${\mathcal C}_{\phi}$ is finite.
\begin{lemma}
The frame $(W_{\phi}^{0},R_{\phi}^{0})$ is in ${\mathcal C}_{\phi}$.
\end{lemma}
\begin{proof}
This is an immediate consequence of the definitions and Lemma~\ref{lemma:3:about:phi:tips}.
\end{proof}
Let $\ll_{\phi}$ be the partial order on ${\mathcal C}_{\phi}$ such that for all $(W,R),(W^{\prime},R^{\prime})\in{\mathcal C}_{\phi}$, $(W,R)\ll_{\phi}(W^{\prime},R^{\prime})$ if and only if $W{\subseteq}W^{\prime}$ and $R\subseteq R^{\prime}$.
\begin{lemma}\label{lemma:well:foundedness}
$({\mathcal C}_{\phi},\ll_{\phi})$ is well-founded.
\end{lemma}
\begin{proof}
This is an immediate consequence of the fact that ${\mathcal C}_{\phi}$ is finite.
\end{proof}
For all $(W,R)\in{\mathcal C}_{\phi}$, let $\sigma_{\phi}(W,R)$ be the relational structure $(W^{\prime},R^{\prime})$ where
\begin{itemize}
\item $W^{\prime}$ is the set of all $(a_{1},\ldots,a_{n_{\phi}})\in W$ such that:\\
for all $i,j\in(n_{\phi})$, if $\psi_{i}=\square\psi_{j}$ and $a_{i}=0$ then there exists $(b_{1},\ldots,b_{n_{\phi}})\in W$ such that $(a_{1},\ldots,a_{n_{\phi}})R(b_{1},\ldots,b_{n_{\phi}})$ and $b_{j}=0$,
\item $R^{\prime}$ is the binary relation on $W^{\prime}$ s.th.\ for all $(a_{1},\ldots,a_{n_{\phi}}),(b_{1},\ldots,b_{n_{\phi}})\in W^{\prime}$, $(a_{1},\ldots,a_{n_{\phi}})R^{\prime}(b_{1},\ldots,b_{n_{\phi}})$ if and only if $(a_{1},\ldots,a_{n_{\phi}})R(b_{1},\ldots,b_{n_{\phi}})$ and there exists $(c_{1},\ldots,c_{n_{\phi}})\in W$ such that $(a_{1},\ldots,a_{n_{\phi}})R(c_{1},\ldots,c_{n_{\phi}})$ and $(c_{1},\ldots,c_{n_{\phi}})R(b_{1},\ldots,b_{n_{\phi}})$.
\end{itemize}
\begin{lemma}\label{lemma:inclusion:sigma}
For all $(W,{R})\in{\mathcal C}_{\phi}$, $\sigma_{\phi}(W,R)\in{\mathcal C}_{\phi}$.
Moreover, $\sigma_{\phi}(W,R)\ll_{\phi}(W,{R})$.
\end{lemma}
\begin{proof}
This is an immediate consequence of the definitions.
\end{proof}
\begin{lemma}
For all $(W,R)\in{\mathcal C}_{\phi}$, there exists $k\in\nat$ such that $\sigma_{\phi}^{k{+}1}(W,R)=\sigma_{\phi}^{k}(W,R)$.
\end{lemma}
\begin{proof}
By Lemmas~\ref{lemma:well:foundedness} and~\ref{lemma:inclusion:sigma}.
\end{proof}
Let $k_{\phi}$ be the least $k\in\nat$ s.th.\ $\sigma_{\phi}^{k{+}1}(W_{\phi}^{0},R_{\phi}^{0})=\sigma_{\phi}^{k}(W_{\phi}^{0},R_{\phi}^{0})$.
Let $(W_{\phi}^{k_{\phi}},R_{\phi}^{k_{\phi}})$ be $\sigma_{\phi}^{k_{\phi}}(W_{\phi}^{0},R_{\phi}^{0})$.
\begin{lemma}
$(W_{\phi}^{k_{\phi}},R_{\phi}^{k_{\phi}}){\models}\Kde$.
\end{lemma}
\begin{proof}
As is well-known, a frame validates $\Kde$ if and only if it is dense.
Hence, it suffices to demonstrate that $(W_{\phi}^{k_{\phi}},R_{\phi}^{k_{\phi}})$ is dense.
For the sake of the contradiction, suppose $(W_{\phi}^{k_{\phi}},R_{\phi}^{k_{\phi}})$ is not dense.
Thus, there exists $(a_{1},\ldots,a_{n_{\phi}}),(b_{1},\ldots,
$\linebreak$
b_{n_{\phi}})\in W_{\phi}^{k_{\phi}}$ such that that $(a_{1},\ldots,a_{n_{\phi}})R_{\phi}^{k_{\phi}}(b_{1},\ldots,b_{n_{\phi}})$ and for all $(c_{1},\ldots,c_{n_{\phi}})
$\linebreak$
\in W_{\phi}^{k_{\phi}}$, either not $(a_{1},\ldots,a_{n_{\phi}})R_{\phi}^{k_{\phi}}(c_{1},\ldots,c_{n_{\phi}})$, or not $(c_{1},\ldots,c_{n_{\phi}})R_{\phi}^{k_{\phi}}(b_{1},\ldots,
$\linebreak$
b_{n_{\phi}})$.
This contradicts the fact that $\sigma_{\phi}^{k_{\phi}{+}1}(W_{\phi}^{0},R_{\phi}^{0})=\sigma_{\phi}^{k_{\phi}}(W_{\phi}^{0},R_{\phi}^{0})$.
\end{proof}
Let $V_{\phi}^{k_{\phi}}$ be a valuation on $(W_{\phi}^{k_{\phi}},R_{\phi}^{k_{\phi}})$ such that for all atoms $p$ and for all $i\in(n_{\phi})$, if $B_{i}=p$ then $V_{\phi}^{k_{\phi}}(p)=\{(a_{1},\ldots,a_{n_{\phi}})\in W_{\phi}^{k_{\phi}}$: $a_{i}=1\}$.
\begin{lemma}
For all $i\in(n_{\phi})$ and for all $(a_{1},\ldots,a_{n_{\phi}})\in W_{\phi}^{k_{\phi}}$:
\[(W_{\phi}^{k_{\phi}},R_{\phi}^{k_{\phi}},V_{\phi}^{k_{\phi}}),(a_{1},\ldots,a_{n_{\phi}})\models\psi_{i}\mbox{ if and only if }a_{i}=1\]
\end{lemma}
\begin{proof}
By induction on $i$.
Let $(a_{1},\ldots,a_{n_{\phi}})\in W_{\phi}^{k_{\phi}}$.
We only show the case where $\psi_{i}=\square\psi_{j}$ for some $j\in(n_{\phi})$ such that $i{>}j$.
Suppose $(W_{\phi}^{k_{\phi}},R_{\phi}^{k_{\phi}},V_{\phi}^{k_{\phi}}),(a_{1},\ldots,
$\linebreak$
a_{n_{\phi}})\models\psi_{i}$ and $a_{i}=1$ are not equivalent.
Hence, either $(W_{\phi}^{k_{\phi}},R_{\phi}^{k_{\phi}},V_{\phi}^{k_{\phi}}),(a_{1},\ldots,
$\linebreak$
a_{n_{\phi}})\models\psi_{i}$ and $a_{i}\not=1$, or $(W_{\phi}^{k_{\phi}},R_{\phi}^{k_{\phi}},V_{\phi}^{k_{\phi}}),(a_{1},\ldots,a_{n_{\phi}})\not\models\psi_{i}$ and $a_{i}=1$.
In the former case, $a_{i}=0$.
Since $\sigma_{\phi}^{k{+}1}(W_{\phi}^{0},R_{\phi}^{0})=\sigma_{\phi}^{k}(W_{\phi}^{0},R_{\phi}^{0})$, $(W_{\phi}^{k_{\phi}},R_{\phi}^{k_{\phi}})$ is $\sigma_{\phi}^{k_{\phi}}(W_{\phi}^{0},R_{\phi}^{0})$ and $\psi_{i}=\square\psi_{j}$, then there exists $(b_{1},\ldots,b_{n_{\phi}})\in W_{\phi}^{k_{\phi}}$ such that $(a_{1},\ldots,a_{n_{\phi}})R_{\phi}^{k_{\phi}}(b_{1},\ldots,b_{n_{\phi}})$ and $b_{j}=0$.
Since $i{>}j$, then by induction hypothesis, $(W_{\phi}^{k_{\phi}},R_{\phi}^{k_{\phi}},V_{\phi}^{k_{\phi}}),(b_{1},\ldots,b_{n_{\phi}})\not\models\psi_{j}$.
Since $(a_{1},\ldots,a_{n_{\phi}})R_{\phi}^{k_{\phi}}(b_{1},\ldots,b_{n_{\phi}})$, then $(W_{\phi}^{k_{\phi}},R_{\phi}^{k_{\phi}},V_{\phi}^{k_{\phi}}),(a_{1},\ldots,a_{n_{\phi}})\not\models\psi_{i}$: a contradiction.
In the latter case, there exists $(b_{1},\ldots,b_{n_{\phi}})\in W_{\phi}^{k_{\phi}}$ such that $(a_{1},\ldots,a_{n_{\phi}})R_{\phi}^{k_{\phi}}(b_{1},\ldots,b_{n_{\phi}})$ and $(W_{\phi}^{k_{\phi}},R_{\phi}^{k_{\phi}},V_{\phi}^{k_{\phi}}),(b_{1},\ldots,b_{n_{\phi}})\not\models\psi_{j}$.
Since $\sigma_{\phi}^{k{+}1}(W_{\phi}^{0},R_{\phi}^{0})=\sigma_{\phi}^{k}(W_{\phi}^{0},R_{\phi}^{0})$, $(W_{\phi}^{k_{\phi}},
$\linebreak$
R_{\phi}^{k_{\phi}})$ is $\sigma_{\phi}^{k_{\phi}}(W_{\phi}^{0},R_{\phi}^{0})$, $a_{i}=1$ and $\psi_{i}=\square\psi_{j}$, then $b_{j}=1$.
Since $i{>}j$, then by induction hypothesis, $(W_{\phi}^{k_{\phi}},R_{\phi}^{k_{\phi}},V_{\phi}^{k_{\phi}}),(b_{1},\ldots,b_{n_{\phi}})\models\psi_{j}$: a contradiction.
\end{proof}
\begin{lemma}\label{lemma:about:truth:in:fixpoint:model}
$\phi\in\Kde$ if and only if for all $(a_{1},\ldots,a_{n_{\phi}})\in W_{\phi}^{k_{\phi}}$, $a_{n_{\phi}}=1$.
\end{lemma}
\begin{proof}
Suppose $\phi\not\in\Kde$.
Hence, by Lindenbaum's Lemma, there exists a prime $\Kde$-theory $s$ such that $\phi\not\in s$~\cite[Lemma~$4.17$]{Blackburn:deRijke:Venema}.
By construction of $(W_{\phi}^{k_{\phi}},R_{\phi}^{k_{\phi}})$, there exists $(a_{1},\ldots,a_{n_{\phi}})\in W_{\phi}^{k_{\phi}}$ such that $\tau_{\phi}(s)=(a_{1},\ldots,a_{n_{\phi}})$.
Since $\phi\not\in s$, then $a_{n_{\phi}}=0$.
\end{proof}
\begin{proposition}
$\Kde$ is in $\EXPTIME$.
\end{proposition}
\begin{proof}
By Lemma~\ref{lemma:about:truth:in:fixpoint:model}, it suffices to demonstrate that given a formula $\phi$, $(W_{\phi}^{k_{\phi}},
$\linebreak$
R_{\phi}^{k_{\phi}})$ can be deterministically computed in exponential time.
As the reader meay easily see, this is a consequence of the fact that $\card(W_{\phi}^{0})\leq2^{\pipe\phi\pipe}$ and for all $(W,R)\in{\mathcal C}_{\phi}$, $\sigma_{\phi}(W,R)$ can be deterministically computed in polynomial time with respect to the size of $(W,R)$.
\end{proof}
\section{A weakly dense logic}
%
%
Let \Kdeab\ be the modal logic $\K_a\oplus \K_b$+$\Ba\Bb p\rightarrow \Ba p$.
Obviously, \Kdeab\ is a conservative extension of ordinary modal logic $\K$.
Hence, \Kdeab\ is $\PSPACE$-hard.\label{pspace-hardness}
In Section~\ref{section:complexity:of:KDeab}, we prove that \Kdeab\ is in $\PSPACE$.
%
%
%
%
%
%
\paragraph{Syntax}
Let $\At$ be the set of all atoms $(p,q,\ldots)$.
The set $\Fo$ of all formulas $(\phi,\psi,\ldots)$ is now defined by
\[\phi:=p\mid\bot\mid\neg\phi\mid(\phi\wedge\phi)\mid\square_{a}\phi\mid\square_{b}\phi\]
where $p$ ranges over $\At$.
As before, we follow the standard rules for omission of the parentheses, we use the standard abbreviations for the Boolean connectives $\top$, $\vee$ and $\rightarrow$ and for all formulas $\phi$, $d(\phi)$ denotes the degree of $\phi$ and $\pipe\phi\pipe$ denotes the number of occurrences of symbols in $\phi$.
For all formulas $\phi$, we write $\lozenge_{a}\phi$ as an abbreviation instead of $\neg\square_{a}\neg\phi$ and we write $\lozenge_{b}\phi$ as an abbreviation instead of $\neg\square_{b}\neg\phi$.
%
%
%
%
\paragraph{Semantics}
A {\em frame}\/ is now a $3$-tuple $(W,R_{a},R_{b})$ where $W$ is a nonempty set and $R_{a}$ and $R_{b}$ are binary relations on $W$.
A frame $(W,R_{a},R_{b})$ is {\em weakly dense}\/ if for all $s,t\in W$, if $sR_{a}t$ then there exists $u\in W$ such that $sR_{a}u$ and $uR_{b}t$.
A {\em valuation on a frame $(W,R_{a},R_{b})$}\/ is a function $V\ :\ \At\longrightarrow\wp(W)$.
A {\em model}\/ is a $4$-tuple consisting of the $3$ components of a frame and a valuation on that frame.
A {\em model based on the frame $(W,R_{a},R_{b})$}\/ is a model of the form $(W,R_{a},R_{b},V)$.
With respect to a model $(W,R_{a},R_{b},V)$, for all $s\in W$ and for all formulas $\phi$, the {\em satisfiability of $\phi$ at $s$ in $(W,R_{a},R_{b},V)$}\/ (in symbols $s\models\phi$) is inductively defined as usual.
In particular,
\begin{itemize}
\item $s\models\square_{a}\phi$ if and only if for all $t\in W$, if $sR_{a}t$ then $t\models\phi$,
\item $s\models\square_{b}\phi$ if and only if for all $t\in W$, if $sR_{b}t$ then $t\models\phi$.
\end{itemize}
As a result,
\begin{itemize}
\item $s\models\lozenge_{a}\phi$ if and only if there exists $t\in W$ such that $sR_{a}t$ and $t\models\phi$,
\item $s\models\lozenge_{b}\phi$ if and only if there exists $t\in W$ such that $sR_{b}t$ and $t\models\phi$.
\end{itemize}
A formula $\phi$ is {\em true in a model $(W,R_{a},R_{b},V)$}\/ (in symbols $(W,R_{a},R_{b},V)\models\phi$) if for all $s\in W$, $s\models\phi$.
A formula $\phi$ is {\em valid in a frame $(W,R_{a},R_{b})$}\/ (in symbols $(W,R_{a},R_{b})\models\phi$) if for all models $(W,R_{a},R_{b},V)$ based on $(W,R_{a},R_{b})$, $(W,R_{a},R_{b},V)\models\phi$.
A formula $\phi$ is {\em valid in a class ${\mathcal C}$ of frames}\/ (in symbols ${\mathcal C}\models\phi$) if for all frames $(W,R_{a},R_{b})$ in ${\mathcal C}$, $(W,R_{a},R_{b})\models\phi$.
\paragraph{A decision problem}
Let $DP_{a,b}$ be the following decision problem:
\begin{description}
\item[input:] a formula $\phi$,
\item[output:] determine whether $\phi$ is valid in the class of all weakly dense frames.
\end{description}
Using the fact that the least filtration of a weakly dense model is weakly dense, one may readily prove that $DP_{a,b}$ is in $\coNEXPTIME$.
We will prove in Section~\ref{section:complexity:of:KDeab} that $DP_{a,b}$ is in $\PSPACE$.
\paragraph{Axiomatization}
In our language, a {\em modal logic}\/ is a set of formulas closed under uniform substitution, containing the standard axioms of $\CPL$, closed under the standard inference rules of $\CPL$, containing the axioms
\begin{description}
\item[$(\Axiom1_{a})$] $\square_{a}p\wedge\square_{a}q\rightarrow\square_{a}(p\wedge q)$,
\item[$(\Axiom2_{a})$] $\square_{a}\top$,
\item[$(\Axiom1_{b})$] $\square_{b}p\wedge\square_{b}q\rightarrow\square_{b}(p\wedge q)$,
\item[$(\Axiom2_{b})$] $\square_{b}\top$,
\end{description}
and closed under the inference rules
\begin{description}
\item[$(\Rule1_{a})$] $\frac{p\rightarrow q}{\square_{a}p\rightarrow\square_{a}q}$,
\item[$(\Rule1_{b})$] $\frac{p\rightarrow q}{\square_{b}p\rightarrow\square_{b}q}$.
\end{description}
Let \Kdeab\ be the least modal logic containing the formula $\square_{a}\square_{b}p\rightarrow\square_{a}p$.
As is well-known, \Kdeab\ is equal to the set of all formulas $\phi$ such that $\phi$ is valid in the class of all weakly dense frames.
This can be proved by using the so-called canonical model construction.
%
%
%
%
%
%
%
%
%
%
\section{Windows}\label{section:windows}
Let $w$ be a finite set of formulas.
We define $d(w)=\max\{d(\phi):\ \phi\in w\}$ and $\pipe w\pipe=\Sigma\{\pipe\phi\pipe:\ \phi\in w\}$.
Moreover, let $\CSF(w)$ be the least set $u$ of formulas such that for all formulas $\phi,\psi$,
\begin{itemize}
\item $w\subseteq u$,
\item if $\phi \wedge \psi\in u$ then $\phi \in u$ and $\psi\in u$,
\item if $\neg (\phi \wedge \psi)\in u$ then $\neg \phi\in u$ and $\neg \psi\in u$,
\item if $\neg \phi\in u$ then $\phi \in u$.
\end{itemize}
In other respect, $\SF(w)$ is the least set $u$ of formulas s. th.\ for all formulas $\phi,\psi$,
\begin{itemize}
\item $w\subseteq u$,
\item if $\phi \wedge \psi\in u$ then $\phi \in u$ and $\psi\in u$,
\item if $\neg (\phi \wedge \psi)\in u$ then $\neg \phi\in u$ and $\neg \psi\in u$,
\item if $\neg \phi\in u$ then $\phi \in u$,
\item if $\square_{a} \phi\in u$ then $\phi \in u$,
\item if $\neg\square_{a} \phi\in u$ then $\neg\phi \in u$,
\item if $\square_{b} \phi\in u$ then $\phi \in u$,
\item if $\neg\square_{b} \phi\in u$ then $\neg\phi \in u$.
\end{itemize}
Finally, let $\Ba^{\mhyphen}(w)=\{\phi\colon \Ba\phi \in w \}$ and $\Bb^{\mhyphen}(w)=\{\phi\colon \Bb\phi \in w \}$.
Notice that $d(\Box_{a}^{\mhyphen}(w))\leq d(w)-1$ and $d(\Box_{b}^{\mhyphen}(w))\leq d(w)-1$.
\\
\\
For all finite sets $u$ of formulas, let $\CCS(u)$ be the set of all finite sets $w$ of formulas such that $u\subseteq w\subseteq \CSF(u)$ and for all formulas $\phi,\psi$,
\begin{itemize}
\item if $\phi \wedge \psi\in w$ then $\phi \in w$ and $\psi\in w$,
\item if $\neg (\phi \wedge \psi)\in w$ then $\neg \phi\in w$ or $\neg \psi\in w$,
\item if $\neg \neg \phi\in w$ then $\phi \in w$,
\item $\bot\not\in w$,
\item if $\neg \phi\in w$ then $\phi\not\in w$.
\end{itemize}
For all finite sets $u$ of formulas, the elements of $\CCS(u)$ are in fact simply unsigned saturated open branches for tableaux of classical propositional logic (see \cite{Smullyan68}).
As a result, for all finite sets $u$ of formulas, an element of $\CCS(u)$ is called a {\em consistent classical saturation (CCS) of $u$.}
%
%
As the reader may easily verify, for all finite sets $u, w$ of formulas, if $w\in \CCS(u)$ then $d(u)=d(w)$ and $\CCS(w)=\{w\}$.
Moreover, there exists an integer $c_{0}$ such that for all finite sets $u, w$ of formulas, if $w\in \CCS(u)$ then $\pipe w\pipe\leq c_{0}.\pipe u\pipe$.
\begin{proposition} [Properties of {\CCS}s]\label{prop-CCS}
For all finite sets $u, v, w, w_1, w_2$ of formulas,
\begin{enumerate}
%
%
%
%
%
%
%
%
\item\label{One} if $w\in \CCS(u\cup w_1)$ and $w_1\in \CCS(v)$ then $w\in \CCS(u\cup v)$,
\item\label{Two} if $w\in \CCS(u\cup v)$ then it exists $v_1\in \CCS(u)$ and $v_2\in\CCS(v)$ s.th.\ \ $v_1\cup v_2=w$,
\item\label{Three} if $w\in \CCS(u\cup w_1)$ and $w_1$ is a \CCS\ then it  exists $v_2\in \CCS(u)$ s.th.\ $w_1\cup v_2 = w$,
%
%
%
%
\item\label{Four} if $w\in \CCS(u\cup w_1)$ and $w_1\in \CCS(v)$ then $d(w\setminus w_1)\leq d(u)$,
\item\label{Five} if $u$ is true at a world $x\in W$ of a \Kdeab-model $M=(W,R_{a},R_{b},V)$, then 
the set $\SF(u)\cap \{\phi\colon M,x\models \phi\}$ is in $\CCS(u)$.
%
%
\end{enumerate}
\end{proposition}
\begin{proof}
%
%
Item~(\ref{One}) is an immediate consequence of the properties of classical open branches of tableaux.
As for Item~(\ref{Two}), take $v_1=w\cap\CSF(u)$ and $v_2=w\cap\CSF(v)$.
Item~(\ref{Three}) follows from Item~(\ref{Two}).
Concerning Item~(\ref{Four}), if $w\in \CCS(u\cup w_1)$ then by Item~(\ref{Three}), there exists $w_2\colon w_2\in \CCS(u)$ and $w_1\cup w_2 = w$.
Therefore, $w\setminus w_1\subseteq w_2$ and $d(w\setminus w_1)\leq d(w_2)=d(u)$.
Finally, about Item~(\ref{Five}), the reader may easily verify it 
by applying the definition clauses of $\models$. 
\end{proof}
%
%
%
%
Let $u$ be a finite set of formulas and $w$ be a \CCS\ of $u$.
Let $k\geq d(w)$.
A {\em $k$-window for $w$}\/ (Fig. \ref{window1}) is a sequence $(w_i)_{0\leq i\leq k}$ of sets of formulas (called {\em dense-successors of $w$}) such that
\begin{enumerate}
\item $w_k\in \CCS(\Ba^{\mhyphen}(w))$,
\item for all $0\leq i < k$, $w_i\in \CCS(\Ba^{\mhyphen}(w)\cup \Bb^{\mhyphen}(w_{i+1}))$.
\item [] (Notice that for all $0\leq i \leq k$, $\lgx{w_i}\leq c_0.\lgx{\Ba^{\mhyphen}(w)\cup \Bb^{\mhyphen}(w_{i+1})}\leq c_0.\lgw$)
\end{enumerate}
An {\em $\infty$-window for $w$}\/ is an infinite sequence $(w_i)_{0\leq i}$ of sets of formulas such that
for all $i\geq0$, $w_i\in \CCS(\Ba^{\mhyphen}(w)\cup \Bb^{\mhyphen}(w_{i+1}))$.

%
%

\noindent 
\framebox{\begin{minipage}{0.97\textwidth}   
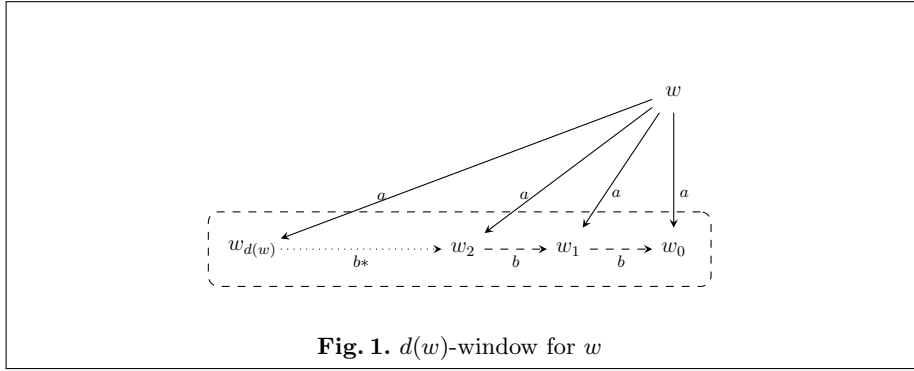
\begin{figure}[H]
\centering
\begin{tikzpicture}[->, >=stealth, scale=0.7, transform shape,font=\large]

\tikzstyle{state}=[minimum size=0.8cm, inner sep=2pt]

\node[state] (w) at (10,3) {$w$};
\node[state] (w0) at (10,0) {$w_0$};
\node[state] (w1) at (8,0) {$w_1$};
\node[state] (w2) at (6,0) {$w_2$};
\node[state] (wd) at (2,0) {$w_{d(w)}$};

\node[draw, dashed, rounded corners, fit=(wd)(w2)(w1)(w0), inner sep=6pt, label=below:{\small }] {};

\draw (w) -- node[pos=0.7, left] {\small$a$} (wd);
\draw (w) -- node[pos=0.7, left] {\small$a$} (w2);
\draw (w) -- node[pos=0.7, right] {\small$a$} (w1);
\draw (w) -- node[pos=0.7, right] {\small$a$} (w0);

\draw[dotted] (wd) -- node[below] {\small$b*$} (w2);
\draw[dashed] (w2) -- node[below] {\small$b$} (w1);
\draw[dashed] (w1) -- node[below] {\small$b$} (w0);

\end{tikzpicture}
\caption{$d(w)$-window for $w$}
\label{window1}
\end{figure}
\end{minipage}
}
%
%
\\
\\
%
%
Let $T_0=(w_i)_{0\leq i\leq k}$ and $T_1=(\wpp_i)_{1\leq i\leq k+1}$ be two $k$-windows for $w$: {\em $T_1$ continues $T_0$ for $w$}\/ if for all $i\in\{1,\ldots,k\}$, $\wpp_i\in \CCS(\Bb^{\mhyphen}(\wpp_{i+1})\cup w_i)$ (Fig.\ \ref{window}).
%
%
%
%
\begin{lemma}[Property of continuations]\label{prop-cont}
Let $u$ be a finite set of formulas and $w$ be a \CCS\ of $u$.
Let $k\geq d(w)$.
Let $T_0=(w_i)_{0\leq i\leq k}$ be a $k$-windows for $w$.
If it exists $T_1=(\wpp_i)_{1\leq i\leq k+1}$ which continues $T_0$ for $w$ then $(w_0,\wpp_1,\wpp_2,\cdots,\wpp_{k+1})$ is a $(k+1)$-window for $w$.
%
%
\end{lemma}
\begin{proof}
First we prove that 
for $1\leq i\leq k\colon d(\wpp_i\setminus w_i) \leq d(w)-k+i-1$
by descending induction on $i\in\{1,\ldots,k\}$.
Take $i\in\{1,\ldots,k\}$.
Then, either $i=k$, or $i<k$.
In the former case, $\wpp_k\in \CCS(\Bb^{\mhyphen}(\wpp_{k+1})\cup w_k)$.
Since $w_k\in \CCS(\Ba^{\mhyphen}(w))$ and $\wpp_{k+1}\in \CCS(\Ba^{\mhyphen}(w))$, then $d(w_k)\leq d(w)-1$ and $d(\Bb^{\mhyphen}(\wpp_{k+1}))\leq d(w)-2$.
Consequently, $d(\wpp_k\setminus w_k) \leq d(w)-1$.
In the latter case, $\wpp_i\in \CCS(\Bb^{\mhyphen}(\wpp_{i+1})\cup w_i)$; and since $\wpp_{i+1}=\wpp_{i+1}\cup w_{i+1}=(\wpp_{i+1}\setminus w_{i+1})\cup w_{i+1}$, and $\Bb^{\mhyphen}(A\cup B)=\Bb^{\mhyphen}(A)\cup \Bb^{\mhyphen}(B)$, we have $\wpp_i\in \CCS(\Bb^{\mhyphen}(\wpp_{i+1}\setminus w_{i+1})\cup \Bb^{\mhyphen}(w_{i+1})\cup w_i)$; but $\Bb^{\mhyphen}(w_{i+1})\subseteq w_i$, hence $\wpp_i\in \CCS(\Bb^{\mhyphen}(\wpp_{i+1}\setminus w_{i+1})\cup w_i)$. 
Now, by Prop.\ \ref{prop-CCS}.\ref{Three}: $\exists u\colon u\in\CCS(\Bb^{\mhyphen}(\wpp_{i+1}\setminus w_{i+1}))$ and $\wpp_i=w_i\cup u$. Thus $\wpp_i\setminus w_i\subseteq u$, and  $d(\wpp_i\setminus w_i)\leq d(u)=d(\Bb^{\mhyphen}(\wpp_{i+1}\setminus w_{i+1}))\leq d(\wpp_{i+1}\setminus w_{i+1})- 1 \leq d(w)-k+i-1$ (by IH). 

Now we check that $(w_0,\wpp_1,\wpp_2,\cdots,\wpp_{k+1})$ is a $k+1$-window for $w$ by examining the definition of continuations.
Firstly, $\wpp_{k+1}\in \CCS(\Ba^{\mhyphen}(w))$.
Secondly, since $\wpp_k\in \CCS(\Bb^{\mhyphen}(\wpp_{k+1})\cup w_k)$ and $w_k\in \CCS(\Ba^{\mhyphen}(w))$, then $\wpp_k\in \CCS(\Bb^{\mhyphen}(\wpp_{k+1})\cup \Ba^{\mhyphen}(w))$.
Thirdly, take $i\in\{1,\ldots,k-1\}$.
Then, $\wpp_i\in \CCS(\Bb^{\mhyphen}(\wpp_{i+1})\cup w_i)$ and $w_i\in \CCS(\Ba^{\mhyphen}(w)\cup\Bb^{\mhyphen}(w_{i+1}))$.
Hence by Prop.\ \ref{prop-CCS}.\ref{One}, $\wpp_i\in \CCS(\Bb^{\mhyphen}(\wpp_{i+1})\cup\Ba^{\mhyphen}(w)\cup\Bb^{\mhyphen}(w_{i+1}))$.
Since $T_1$ is a continuation of $T_0$, $w_{i+1}\subseteq \wpp_{i+1}$. Then $\Bb^{\mhyphen}(w_{i+1})\subseteq \Bb^{\mhyphen}(\wpp_{i+1})$, and $\wpp_i\in \CCS(\Bb^{\mhyphen}(\wpp_{i+1})\cup\Ba^{\mhyphen}(w))$.
Fourthly, it remains to prove that $w_0\in \CCS(\Bb^{\mhyphen}(\wpp_1)\cup \Ba^{\mhyphen}(w))$.
By the fact above, $d(\wpp_1\setminus w_1)\leq d(w)-k\leq 0$, hence if $\Bb \phi\in\wpp_1$ then $\Bb \phi\in w_1$
and thus $\Bb^{\mhyphen}(\wpp_1)=\Bb^{\mhyphen}(w_1)$.
Since $w_0\in \CCS(\Bb^{\mhyphen}(w_1)\cup \Ba^{\mhyphen}(w))$, then $w_0\in \CCS(\Bb^{\mhyphen}(\wpp_1)\cup \Ba^{\mhyphen}(w))$. 
\end{proof}
\vspace{-0.3cm}

\begin{lemma}[Loop and existence of infinite window]\label{corollary}
Let $u$ be a finite set of formulas and $w$ be a \CCS\ of $u$.
Let $(T_i)_{0\leq i\leq 2^{c_{0}.(d(w)+1).\lgw}}$ be a sequence of $d(w)$-windows for $w$ such that for all $i< 2^{c_{0}.(d(w)+1).\lgw}$, $T_{i+1}$ is a continuation of $T_i$ for $w$.
Then there exists a $\infty$-window for $w$. 
\end{lemma}
\begin{proof}
All \CCS\ used in $d(w)$-windows for $w$ have their size  bounded by $c_{0}.\lgw$, then there are at most $2^{c_{0}.(d(w)+1).\lgw}$ distinct $d(w)$-windows for $w$.
Hence, there exists integers $h,\delta$ such that $\delta\neq 0$ and $h+\delta \leq 2^{c_{0}.(d(w)+1).\lgw}$ and $T_h=T_{h+\delta}$.
\\
Let $(\Tpp_i)_{0\leq i}$ be the infinite sequence such that for all $i\leq h$, $\Tpp_i=T_i$ and for all $i>h$, $\Tpp_i=T_{h+((i-h)\!\!\!\mod \delta)}$.
By construction, for all $i\geq 0$, $\Tpp_{i+1}$ is a continuation of $\Tpp_i$ for $w$.
For all $i\geq 0$, suppose that $\Tpp_i=(w^i_0,\cdots,w^i_{d(w)})$.
For all $i\geq0$, let $\wpp_i=w^i_0$.
As the reader may easily verify, $(\wpp_i)_{i\geq 0}$ is an infinite window for $w$ .
\end{proof}
\section{The algorithm}
Because of Prop. \ref{prop-CCS}.\ref{Five}, testing the $\Kdeab$-satisfiability of a set $u$ of formulas amouts to testing that of a $\CCS$, since $u$ is $\Kdeab$-satisfiable if and only if there exists a $\Kdeab$-satisfiable $w\in\CCS(u)$. Hence, given an initial set of formulas $u$ to be tested, the initial call is $\sat(\chooseCCS(\{u\}))$.\\
In what follows we use built-in functions \algand\ and \all.
The former function lazily implements a logical ``and".
The latter function lazily tests if all members of its list argument are true. \\
\setlength{\textfloatsep}{0pt}
\setlength{\floatsep}{0pt}
\vspace{-1cm}
\begin{algorithm}[H]
 \floatname{algorithm}{Function}
\begin{algorithmic}
\caption{Test for \Kdeab-satisfiability of a set $w$: $w$ must be classically consistent and recursively each $\Diamond$-formula must be satisfied as well as all the dense-successors of $w$.}
\Function{\sat}{$w$}:
\State {return}
\State {\hspace{0.87cm}$w\neq \{\bot\}$}
\State {\hspace{0.2cm}\algand\ \all $\{$\sat(\chooseCCS$(\{\neg\phi\}\cup \Bb^{\mhyphen}(w))\colon \neg\Bb\phi\in w\}$}
\State {\hspace{0.2cm}\algand\ \all $\{\satW(\chooseW(w,\neg\phi),w,\uplim)\colon\neg\Ba\phi\in w\}$}
\EndFunction
\end{algorithmic}
\end{algorithm}
\vspace{-1cm}

\begin{algorithm}[H]
\floatname{algorithm}{Function}
\begin{algorithmic}
\caption{Returns $\{\bot\}$ if $x$ is not classically consistent, otherwise returns one classically saturated open branch non-deterministically choosen}
\Function{\chooseCCS}{$x$}
\If  {$\CCS(x)\neq \emptyset$}
\State {return one $w\in \CCS(x)$}
\Else 
\State {return $\{\bot\}$}
\EndIf
\EndFunction
\end{algorithmic}
\end{algorithm}
\vspace{-1.5cm}

\begin{algorithm}[H]
\floatname{algorithm}{Function}
\begin{algorithmic}
\caption{Non-deterministically chooses a $d(w)$-window for $w$ if possible, see fig.\ \ref{window1}}
\Function{\chooseW}{$w$,$\neg\phi$}
\If {there exists a $d(w)$-window $(w_0,\cdots,w_{d(w)})$ for $w$ such that $\neg \phi \in w_0$}
\State {return $(w_0,\cdots,w_{d(w)})$}
\Else 
\State {return $(\{\bot\},\cdots,\{\bot\})$}
\EndIf
\EndFunction
\end{algorithmic}
\end{algorithm}
\vspace{-1.5cm}

\begin{algorithm}[H]
\floatname{algorithm}{Function}
\begin{algorithmic}
\caption{Tests the satisfiability of each dense-successor of a window for $w$ and recursively for those of its continuation until a repetition happens or a contradiction is detected}
\Function{\satW}{$((w_0,\cdots,w_{d(w)})$,$w$,$N$}:
\If {$N=0$}
\State {return \true}
\Else
\State {return}
\State {\hspace{0.87cm} \sat$(w_0)$}
\State {\hspace{0.2cm} \algand\ \satW(\nextW$((w_0,\cdots,w_{d(w)}),w),w,N-1)$}
\EndIf
\EndFunction
\end{algorithmic}
\end{algorithm}

\begin{algorithm}[H]
\floatname{algorithm}{Function}
\begin{algorithmic}
\caption{Non-deterministically chooses a continuation of a window for $w$ if possible, see fig.\ \ref{window}}
\Function{\nextW}{$T_0=(w_0,\cdots,w_{d(w)})$,$w$}
\If {there exists  a continuation $T_1$ of $T_0$ for $w$}
\State {return $T_1$}
\Else 
\State {return $(\{\bot\},\cdots,\{\bot\})$}
\EndIf
\EndFunction
\end{algorithmic}
\end{algorithm}

\noindent \framebox{\begin{minipage}{0.97\textwidth}   
\begin{figure}[H]
\centering
\begin{tikzpicture}[->, >=stealth, scale=0.7, transform shape,font=\large]

\tikzstyle{state}=[minimum size=0.8cm, inner sep=2pt]

\node[state] (w) at (10,3) {$w$};
\node[state] (w0) at (10,0) {$w_0$};
\node[state] (w1) at (8,0) {$w_1$};
\node[state] (w2) at (6,0) {$w_2$};
\node[state] (wd) at (2,0) {$w_{d(w)}$};
\node[state] (wd1) at (0,0) {$w_{d(w)+1}$};

\node[draw, dashed, rounded corners, fit=(wd1)(wd)(w2)(w1), inner sep=6pt, label=below:{\small Next $d(w)$-window for $w$} once $\sat(w_0)$ has returned \true] {};

\draw (w) -- node[pos=0.7, left] {\small$a$} (wd1);
\draw (w) -- node[pos=0.7, left] {\small$a$} (wd);
\draw (w) -- node[pos=0.7, left] {\small$a$} (w2);
\draw (w) -- node[pos=0.7, right] {\small$a$} (w1);
\draw (w) -- node[pos=0.7, right] {\small$a$} (w0);

\draw[dashed] (wd1) -- node[below] {\small$b$} (wd);
\draw[dotted] (wd) -- node[below] {\small$b*$} (w2);
\draw[dashed] (w2) -- node[below] {\small$b$} (w1);
\draw[dashed] (w1) -- node[below] {\small$b$} (w0);

\end{tikzpicture}
\caption{Results of $\nextW$}
\label{window}
\end{figure}
\end{minipage}
}\mbox{}\\

\section{Analysis of the algorithm}\label{section:complexity:of:KDeab}

Given a \Kdeab-model $M=(W,\Ra,\Rb,v)$ and a set $s$ of formulas, we will write $M,x\models s$ for $\forall \phi\in s\colon M,x\models \phi$. 

\begin{lemma}[Soundness]\label{soundness}\\
If $w$ is a \Kdeab-satisfiable \CCS\ then the call \sat(w) returns \true. 
\end{lemma}
\begin{proof} 
Since $w$ is \Kdeab-satisfiable, then $w\neq \{\bot\}$. Hence the result of $\sat(w)$ rely on that of: \\
$\all \{\sat(\chooseCCS(\{\neg\phi\}\cup \Bb^{\mhyphen}(w))\colon \neg\Bb\phi\in w\}$\\ 
\indent $\algand$\\
    $\all\{\satW(\chooseW(w,\neg\phi,d(w)),\Ba^{\mhyphen}(w),\uplim)\colon\neg\phi\in w\}$\\
We proceed by induction on $d(w)$. 
\begin{itemize}
    \item Case $d(w)=0$: then the sets\\
    $\{\sat(\chooseCCS(\{\neg\phi\}\cup \Bb^{\mhyphen}(w))\colon \neg\Bb\phi\in w\}$ and\\
    $\{\satW(\chooseW(w,\neg\phi,d(w)),\Ba^{\mhyphen}(w),\uplim)\colon\neg\phi\in w\}$\\
    are empty. Hence $\sat(w)$ returns \true.
    \item Case $d(w)\geq 1$: for some \Kdeab-model $M=(W,\Ra,\Rb,v)$ and $x\in W$,
    $M,x\models w$ and
    \begin{enumerate}
        \item since  $M,x\models w$ then for all $\neg\Bb\phi\in w$, $M,x\models \neg\Bb\phi$. Hence for all $\neg\Bb\phi\in w$, there exists $y\in W$ s.th.\ $(x,y)\in\Rb$ and $M,y\models \neg\phi$ and $M,y\models \Bb^{\mhyphen}(w)$. Thus for all $\neg\Bb\phi\in w$, if $u_0 =\{\neg\phi\}\cup \Bb^{\mhyphen}(w)$ then $u_0$ is \Kdeab-satisfiable. Let $w_0=\CSF(u_0)\cap y$, then by Prop.\ \ref{prop-CCS}.\ref{Five} $w_0\in\CCS(u_0)$ and $w_0$ is \Kdeab-satisfiable too. Thus by IH (since $d(w_0)<d(w)$), for all $\neg\Bb\phi\in w$ there exists $w_0 \in \CCS(u_0)$ such that $\sat(w_0)$ returns \true. Hence $\{\sat(\chooseCCS(\{\neg\phi\}\cup \Bb^{\mhyphen}(w))\colon \neg\Bb\phi\in w\}$ returns \true.
         \item since  $M,x\models w$ then for all $\neg\Ba\phi\in w$, $M,x\models \neg\Ba\phi$. Hence, for all $\neg\Ba\phi\in w$, there exists an infinite sequence $(y_i)_{i\geq 0}$ such that for $0\leq i$:
         \begin{itemize}
             \item $(x,y_i)\in\Ra$ 
             \item $(y_{i+1},y_i)\in\Rb$ 
             \item $M,y_0\models \neg\phi$
             \item $M,y_i\models \Ba^{\mhyphen}(w)$ 
             \item $M,y_i\models \Bb^{\mhyphen}(y_{i+1})$ 
         \end{itemize}
         Let 
         \begin{itemize}
             \item $w_{\uplim}=\CSF(\Ba^{\mhyphen}(w))\cap y_{\uplim}$
             \item $w_i=\CSF(\Ba^{\mhyphen}(w)\cup\Bb^{\mhyphen}(w_{i+1}))\cap y_i$ for $0\leq i< \uplim$
             \item $w_0=\CSF(\{\neg\phi\}\cup \Ba^{\mhyphen}(w)\cup\Bb^{\mhyphen}(w_{1}))\cap y_0$
         \end{itemize}        
         By Prop.\ \ref{prop-CCS}.\ref{Five}, these $(w_i)_{0\leq i\leq \uplim}$ form a sequence of \Kdeab-satisfiable \CCS\ such that:
         \begin{itemize}
             \item $w_{\uplim}\in \CCS(\Ba^{\mhyphen}(w))$
             \item $w_i\in \CCS(\Ba^{\mhyphen}(w)\cup\Bb^{\mhyphen}(w_{i+1})$ for $1\leq i< \uplim$
             \item $w_0\in\CCS(\{\neg\phi\}\cup \Ba^{\mhyphen}(w)\cup\Bb^{\mhyphen}(w_{1}))$
         \end{itemize}
         Since $d(w_i)<d(w)$ for each $0\leq i\leq \uplim$, then, by IH, $\sat(w_i)$ returns $\true$ for all $0\leq i\leq \uplim$. 
         
        Obviously each subsequence $(w_i,\cdots,w_{i+d(w)})$ is a $d(w)$-window for $w$ and  $(w_{i+1},\cdots,w_{i+d(w)+1})$ is a continuation of it.  Thus for each $\neg\Ba\phi\in w$ the call 
        $\satW(\chooseW(w,\neg\phi),\Ba^{\mhyphen}(w),\uplim)$
        will reduce to returning:
\[\sat(w_0) \mbox{ \algand\ }\sat(w_1) \mbox{ \algand\ } \ldots \mbox{ \algand\ } \sat(w_{\uplim})\]
         which is \true.
\end{enumerate}
\end{itemize}
\end{proof}

\begin{lemma}[Completeness]\label{completeness}\\
Given a set $x$ of formulas, if \sat(\chooseCCS(x)) returns \true, then $x$ is \Kdeab-satisfiable. 
\end{lemma}
\begin{proof}
    We inductively construct $M=(W,\Ra,\Rb,V)$. Let $w\in\CCS(x)$ such that $\sat(w)$ returns \true, and let $\neg\Bb\phi_1,\cdots,\neg\Bb\phi_k$ and $\neg\Ba\phi_1,\cdots,\neg\Ba\phi_l$ the $\Diamond$-formulas of $w$. In what follows we define $V_w$ by: $V_w(p)=\{w\}$ if $p\in w$ and else $V_w(p)=\emptyset$, for all $p\in\At$. 
    \begin{itemize}
        \item Firstly, for $1\leq i\leq k$ and $\neg\Bb\phi_i\in w$ suppose that
        $\sat(\chooseCCS(\{\neg\phi_i\}\cup \Bb^{\mhyphen}(w)))$ returns \true. By IH, $\{\neg\phi_i\}\cup \Bb^{\mhyphen}(w)$ is true in some \Kdeab\-model: $(U_i,\rho a_i,\rho b_i,\alpha_i),v_i\models \{\neg\phi_i\}\cup \Bb^{\mhyphen}(w)$
        \item Secondly, for each $k+1\leq i\leq l$ and for each $\neg\Ba\phi_i\in w$ suppose that the call $\satW(\chooseW(w,\neg\phi_i),w,\uplim)$ returns \true\\
        Then, call to \chooseW\ followed by repeated calls to \nextW\ will compute a sequence of $\uplim$ $d(w)$-windows for $w$ each being a continuation of the previous. Thus by lemma \ref{corollary}, there exists a $\infty$-window for $w$. Let it be denoted by $(\wpp_j^i)_{j\geq 0}$. 
        By IH, each $\wpp_j^i$ is is true in some \Kdeab\-model: $(W_j^i,ra_j^i,rb_j^i,\beta_j^i),v^i\models \wpp_j^i$. 
        \end{itemize}
        Now we set:
        \begin{itemize}
            \item $W=\{w\}\cup\bigcup_{1\leq i\leq k} U_i\cup\bigcup_{k+1\leq i\leq l} \bigcup_{0\leq j} W_j^i$
            \item $\Ra=\bigcup_{1\leq i\leq k} \rho a_i\cup\bigcup_{k+1\leq i\leq l} \bigcup_{0\leq j}(ra_j^i\cup\{(w,\wpp_j^i))\})$
            \item $\Rb=\{(w,v_i)\colon 1\leq i\leq k\}\cup\bigcup_{1\leq i\leq k} \rho b_i\cup\bigcup_{k+1\leq i\leq l} \bigcup_{0\leq j}(rb_j^i\cup\{(\wpp_{j+1},\wpp_j)\})$
            \item $V(p)=\bigcup_{1\leq i\leq k}\alpha_i(p)\cup\bigcup_{k+1\leq i\leq l} \bigcup_{0\leq j}\beta_j^i(p)\cup V_w(p)$ for all $p\in \At$.
        \end{itemize}
        
    Let us show that this Kripke model $M$ is a \Kdeab-model: if $(w,\wpp)\in\Ra$, then  $\exists i\colon k+1\leq i\leq l\colon \exists j\colon (w,\wpp)\in ra_j^i$ or $\wpp=\wpp_j^i$. In the former case there exists $z\colon (w,z)\in\Ra$ and $(z,\wpp)\in\Rb$ because $(W_j^i,ra_j,rb_j)$ is a \Kdeab-frame and in the latter case by definition of $d(w)$-windows for $w$: $(w,\wpp_j^{i+1})\in Ra$ and $(\wpp_j^{i+1},\wpp_j^i)\in Rb$.\\
   Let us prove that for all $v\in W$ and for all $\phi \in v$, then $M,v\models \phi$. We only consider the cases $\neg\Ba\phi,\neg\Bb\phi,\Ba\phi$ and $\Bb\phi$: 
        \begin{itemize}  
            \item For $\neg\Ba\phi_i\in v$ (resp.\  $\neg\Bb\phi_i\in v$), there exists an infinite window $(\wpp_j^i)_{j\geq 0}$ for $v$ (resp.\ there exists a $\CCS\ \wpp_0^i$) with $\neg\phi_i\in \wpp_0^i$ and s.th.\ $\sat(\wpp_0^i)$ returns \true. Thus by IH $M,\wpp_0^i\models \neg\phi_i$. By construction $(v,\wpp_0^i)\in\Ra$ (resp.\ $\in\Rb)$, hence $M,v\models \neg\Ba\phi_i$ (resp.\ $\neg\Bb\phi_i)$.
            \item  For all  $\Ba\phi\in v$ and $(v,\wpp)\in\Ra$, by construction $\wpp$ is some \CCS\ s.th.\ $\wpp\supseteq\Ba^{\mhyphen}(v)$ and by IH $M,\wpp\models \phi$. Hence $M,v\models \Ba\phi$ for all $\Ba\phi\in w$. 
            \item For all  $\Bb\phi\in v$ and $(v,\wpp)\in\Rb$, by construction $\wpp$ is some \CCS\ s.th.\ $\Bb^{\mhyphen}(v)\inc \wpp$ (in the case $\wpp$ comes from the development of a $\neg\Bb$-formula), or it is some $\wpp_j$ of a window and $v$ is $\wpp_{j+1}$, and again $\Bb^{\mhyphen}(v)\inc \wpp$. Thus by IH $M,\wpp\models \phi$. Hence $M,v\models \Bb\phi$ for all $\Bb\phi\in v$. 
            \end{itemize}
\end{proof}

\begin{figure}[htbp]
\centering
\begin{tikzpicture}[
    bullet/.style={circle, fill=black, inner sep=1.2pt},
    -, >=Stealth, font=\large,scale=0.9]

\node[bullet] (n1) at (0,0) {};
\node (u) at (0.3,0) {$u$};

\node[bullet] (n7) at (3.4,-1.5) {};
\node[bullet] (n71) at (3.4,-1.5) {};
\node[bullet] (n6) at (2.6,-1.5) {};
\node[bullet] (n6.1) at (2.8,-1.5) {};
\node[bullet] (n6.2) at (3,-1.5) {};
\node[bullet] (n5) at (2.2,-1.5) {};
\node[bullet] (n4) at (1.4,-1.5) {};
\node[bullet] (n4.1) at (1.6,-1.5) {};
\node[bullet] (n4.2) at (1.8,-1.5) {};

\node[bullet] (n2) at (-2,-1.5) {};
\node[bullet] (n3) at (0,-1.5) {};

\node[bullet] (n10)  at (0,-3) {};
\node[bullet] (n9)  at (-1,-3) {};
\node[bullet] (n8) at (-2,-3) {};
\node[bullet] (n11) at (0.9,-3) {};
\node[bullet] (n11.1) at (0.2,-3) {};
\node[bullet] (n11.2) at (0.4,-3) {};
\node[bullet] (n12) at (1.4,-3) {};
\node[bullet] (n12.1) at (1.6,-3) {};
\node[bullet] (n12.2) at (1.8,-3) {};
\node[bullet] (n13) at (2.2,-3) {};

\node[bullet] (n14) at (0,-4.5) {};
\node (w) at (0.3,-4.5) {$w$};

\node[bullet] (n17) at (0,-6) {};
\node[bullet] (n17.1) at (0.2,-6) {};
\node[bullet] (n17.2) at (0.4,-6) {};
\node (w0) at (0,-6.5) {$w_0$};

\node[bullet] (n16) at (-1,-6) {};
\node[bullet] (n15) at (-1.8,-6) {};
\node[bullet] (n15.1) at (-1.6,-6) {};
\node[bullet] (n15.2) at (-1.8,-6) {};
\node[bullet] (n18) at (0.9,-6) {};
\node (wk) at (1.2,-6.5) {$w_{d(w)}$};

\node[bullet] (n19) at (2,-6) {};
\node[bullet] (n19.1) at (2.2,-6) {};
\node[bullet] (n19.2) at (2.4,-6) {};
\node[bullet] (n20) at (2.8,-6) {};
\node[bullet] (n21) at (1.6,-6) {};

\draw[sloped] (n1) to[bend right=20](n2) ;
\foreach \target in {4,4.1,4.2,5,6,6.1,6.2,71} {
  \draw[densely dotted][sloped] (n1) to[bend left=20](n\target);
}
\draw (n1)   to[bend left=10](n3) ;

\foreach \target in {8,9} {
  \draw[sloped] (n3) to[bend right=20] (n\target);
}
\foreach \target in {12,12.1,12.2,13} {
\draw[sloped][densely dotted] (n3) to[bend left=20] (n\target);
}
\draw [dotted](n3)   to[bend left=10](n10) ;
\foreach \target in {11,11.1,11.2} {
\draw [sloped][densely dotted](n3)   to[bend left=10](n\target);
}
\draw[dashed] (n10) -- (n14) node[midway, left] {};

\foreach \target in {15,15.1,15.2,16} {
  \draw[densely dotted][sloped] (n14) to[bend right=20](n\target) ;
}
\foreach \target in {19,19.1,19.2,20} {
  \draw[densely dotted][sloped] (n14) to[bend left=20](n\target) ;
}

\draw [densely dotted](n14)   to[bend left=10](n17) ;
\foreach \target in {17.1,17.2,18} {
\draw [densely dotted](n14)   to[bend left=10](n\target);
}

\draw[] (n14) to[bend left=20] (n21);


\node[fill=black!10, draw, rounded corners, fit=(n4)(n5), inner sep=2pt, label=below:] {};
\node[fill=black!10, draw, rounded corners, fit=(n6)(n7), inner sep=2pt, label=below:] {};
\node[draw, rounded corners, fit=(n10)(n11), inner sep=2pt, label=below:] {};
\node[fill=black!10, draw, rounded corners, fit=(n12)(n13), inner sep=2pt, label=below:] {};
\node[draw, rounded corners, fit=(n17)(n18), inner sep=2pt, label=below:] {};
\node[fill=black!10, draw, rounded corners, fit=(n19)(n20), inner sep=2pt, label=below:] {};
\node[fill=black!10, draw, rounded corners, fit=(n15)(n16), inner sep=2pt, label=below:] {};

\node[draw, dotted, rounded corners, fit=(n1)(n3)(n10)(n14)(n18), inner sep=5pt, label=below:] {};

\node[bullet] (x14) at (5,-4.5) {};
\node (w) at (5.3,-4.5) {$w$};

\node[bullet] (x171) at (5.2,-6) {};
\node[bullet] (x172) at (5.4,-6) {};
\node[bullet] (x173) at (5.6,-6) {};
\node (wp1) at (5.3,-6.4) {$\wpp_1$};

\node[bullet] (x18) at (6.2,-6) {};
\node (wpk1) at (6.6,-6.4) {$\wpp_{d(w)+1}$};

\foreach \target in {171,172,173,18} {
\draw [densely dotted](x14)   to[bend left=10](x\target);
}

\node[draw, rounded corners, fit=(x171)(x172)(x173)(x18), inner sep=2pt, label=below:] {};
\node (ww0) at (4.7,-6.4) {$w_0$};
\node (cross0) at (4.8,-6) {$\times$};
\node[bullet] (x17) at (4.8,-6) {};
\draw [densely dotted](x14)  to [bend left=10]node[midway]{$\times$} (x17)  ;
\node[] (f2) at (4.6,-5.3) {};
\node[] (f1) at (2.8,-5.3) {};
\draw[line width=0.5mm,double,->](f1) -- (f2) node[midway, below] {$\nextW$};

\end{tikzpicture}
\caption{A view of the computation tree of $\sat(u)$ when has just been executed a call $\satW((w_0,\cdots,w_{d(w)}),w,\uplim)$. Solid lines are $b$-edges, dotted ones are $a$-edges. Small boxes are windows. The big dotted box shows the part stored in memory. On the right, $(\wpp_1,\wpp_2,\cdots,\wpp_{d(w)+1})$ is a  continuation of $(w_0,\cdots,w_{d(w)})$ for $w$, which will be explored once $\sat(w_0)$ will have returned \true\ ($w_0$ can be forgotten).\\ }
\label{memory}
\end{figure}
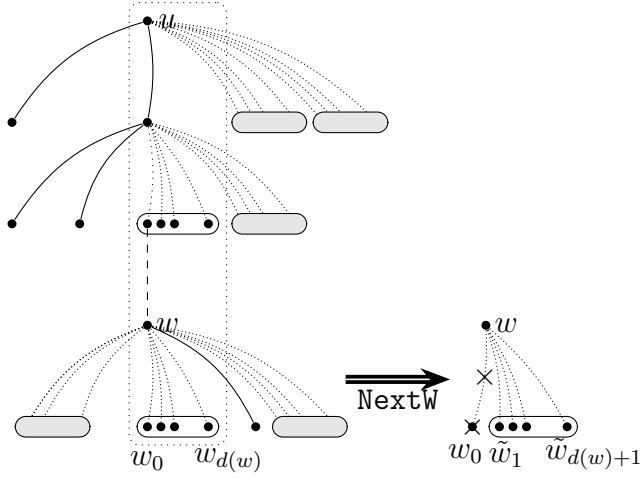
\begin{lemma}
$\sat(w)$ runs in polynomial space w.r.t.\ $\lgw$.    
\end{lemma}
Fig.\ \ref{memory} is provided in order to illustrate how works the algorithm. 
\begin{proof}
    First, we recall that functions \all\ and \algand\ are lazily evaluated. \\
    Obviously, \chooseCCS\ runs in polynomial space. 
    On another hand, the size of each $d(w)$-window for $w$ is bounded by $d(w).\lgw$, hence by $\lgw^2$ since $d(w)\leq \lgw$. Thus the functions \chooseW\ and \nextW\ run in polynomial space, namely $\mathcal{O}(\lgw^2)$. It is also clear that functions $\sat$ and $\satW$ terminate since their recursion depth is bounded (respectively by $\lgw$ and $\lgN$) as well as their recursion width.
    Among all of these calls, let $\wpp$ be the argument for which $\sat$ has the maximum cost in terms of space, i.e.\ such that $space(\sat(\wpp)$ is maximal.\\ 
    Let $T_0=(w_0,\cdots,w_{d(w)})$ be a $d(w)$-window for $w$. Let us firstly evaluate the cost of $space(\satW(T_0,w,N))$. For $0\leq i<N$, let $T_{i+1}$ be the result of $\nextW(T_i,w)$ (note that $\lgx{T_i}=\lgx{T_0}$). The function $\satW$  keeps its arguments in memory during the call $\sat(w_0)$ and either terminate or forget them and continue, hence: \\
    $space(\satW(T_0,w,N))$ \\
    $\begin{array}{lll}
    \leq& \max \{&\lgx{T_0}+\lgw+\lgN+space(\sat(w_0)),\\
    &&space(\satW(T_1,w,N-1))\}\\
    \leq& \max \{&\lgx{T_0}+\lgw+\lgN+space(\sat(\wpp)),\\
    &&\lgx{T_1}+\lgw+\lgx{N-1}+space(\satW(T_2,w,N-1))\}\\
    \leq& \max \{&\lgx{T_0}+\lgw+\lgN+space(\sat(\wpp)),\\
    &&\lgx{T_0}+\lgw+\lgx{N-1}+space(\sat(\wpp)),\\
    &&\cdots\\
    &&\lgx{T_0}+\lgw+\lgx{0}\}\\
    
     \leq&& \lgx{T_0}+\lgw+\lgN+space(\sat(\wpp))
     \end{array}$
    
    \noindent Since $N\leq \uplim$ and $\lgx{T_0}\leq \lgw^2$, $space\satW(T_0,w,N)$ is bounded by $c'.\lgw^2+space(\sat(\wpp))$ for some constant $c'>0$. \\
    Now, concerning the function $\sat$, it also keeps of its argument in memory during recursion in order to range over its $\Diamond$-formulas. Thus:\\
    $space(\sat(w))$\\
$\begin{array}{ll}
    &  \leq \lgw+\max \{space(\sat(\wpp)),c'.\lgw^2+space(\sat(\wpp))\}\\
    & \leq (c'+1).\lgw^2+space(\sat(\wpp))
    \end{array}$
    
With respect to the size of the arguments (and since $\lgx{\wpp}\leq \lgw$) we are left with a recurrence equation of the form: 
    $space(\lgw)\leq space(\lgw-1)+(c'+1).\lgw^2$ with $space(0)=1$ which yields $space(\sat(\lgw))={\mathcal O}(\lgw^3)$.
\end{proof}
\begin{theorem}
    $DP_{a,b}$ is $\PSPACE$-complete. 
\end{theorem}
\begin{proof}
    On the one hand, $DP_{a,b}$ is $\PSPACE$-hard as recalled p.\ \pageref{pspace-hardness}, and on the other hand, our function $\sat$ can decide non-deterministically and within polynomial space whether a set of formulas is satisfiable, $\Kdeab$-satisfiability is in $\NPSPACE$, i.e.\ in $\PSPACE$ (by Savitch' theorem. Thus $DP_{a,b}$ is in co-$\PSPACE$ which is equal to $\PSPACE$e. 
\end{proof}

\section*{Conclusion}
\begin{maybePrint}{festschrift}
Andreas Herzig is one of the main founding members of the Toulouse school of non-classical logics and, in particular, of semantic tableaux for modal logics starting from \cite{FI-tableaux} up to  \cite{Lotrec} and even more. We would have liked to bring a definitive response to the question of the satisfiability problem of all density logics as a contribution to this festschrift for Andreas, and as a tribute to this school. Indeed, we shed a little light over them but despite their apparent simplicity, the exact complexity of density as well as that of multimodal logics with more complex weak forms of density will remain an open problem. And we like open problems !\end{maybePrint}

\begin{maybePrint}{jolli}
Despite their apparent simplicity, the exact complexity of satisfiability of density as well as that of multimodal logics with more complex weak forms of density remains open in full generality. Still, we showed that both selective filtration and our \emph{windows} could provide tools for studying some of them and bring some new results. We think they could be used to solve some of these open problems. 
\end{maybePrint}


\begin{thebibliography}{99}
\bibliographystyle{plain}

\bibitem{Blackburn:deRijke:Venema} P. Blackburn, M. de Rijke, Y. Venema.  Modal logic,, Cambridge Tracts in Theoretical Computer Science - Series, Cambridge Univ. Press, 2001. DOI: 10.1017/CBO9781107050884

\bibitem{FI-tableaux} M. Castilho, Fari\~nas del Cerro, O. Gasquet, A. Herzig,  (1997). Modal Tableaux with Propagation Rules and Structural Rules. \emph{Fundamenta Informatic\ae}: 32. DOI: 10.3233/FI-1997-323404

\bibitem{FAR-GAS99} L. Fari\~nas del Cerro and O. Gasquet. Tableaux Based Decision Procedures for Modal Logics of Confluence and Density. \emph{Fundamenta Informatic\ae}:40, 1999. DOI: 10.3233/FI-1999-40401

\bibitem{FarinasdelCerro1988-FARGL} L. Fari\~nas del Cerro, M. Penttonen. Grammar Logics. \emph{Logique \& Analyse}:31, 1988

\bibitem{Lotrec} O. Gasquet, A. Herzig, B. Said, F. Schwarzentruber. Kripke's Worlds-An Introduction to Modal Logics via Tableaux. Studies in Universal Logic - Series, Springer-Verlag, pp.XV, 198, 2014. DOI: 10.1007/978-3-7643-8504-0

\bibitem{Ladner77} R. E. Ladner. The Computational Complexity of Provability in Systems of Modal Propositional Logic. in \emph{SIAM Journal on Computing}:6(3) 1977. DOI: 10.1137/0206033

\bibitem{Lyon24} T. Lyon, P. Ostropolski-Nalewaja, Decidability of Quasi-Dense Modal Logics, \emph{LICS'24: Proc.\ of the 39th ACM/IEEE Symposium on Logic in Computer Science}, 2024. DOI: 10.1145/3661814.3662111

\bibitem{Smullyan68}  R. M. Smullyan, First-order logic, Berlin, Springer-Verlag, 1968

\end{thebibliography}
\end{document}